\documentclass[preprint,12pt]{elsarticle}

\usepackage{amssymb,amsmath}
\usepackage{multirow}
\usepackage{algorithm}
\usepackage{algorithmic}
\usepackage{color}
\usepackage{graphicx}

\newtheorem{theorem}{Theorem}
\newtheorem{lemma}{Lemma}

\newtheorem{definition}{Definition}
\newdefinition{rmk}{Remark}
\newproof{proof}{Proof}
\newcommand{\be}{\begin{equation}}
\newcommand{\ee}{\end{equation}}
\def\bea{\begin{eqnarray}}
\def\eea{\end{eqnarray}}

\newcommand{\bbmat}{\left\{ \begin{array}}
\newcommand{\ebmat}{\end{array} \right. }

\def\ln{{\rm ln}}

\def\b0{\boldsymbol{\rm 0}}

\begin{document}

\begin{frontmatter}



\title{Streaming Algorithms for the $k$-Submodular Cover Problem}

\author[mymainaddress]{Wenqi Wang}
\ead{wenqiwangcc@nnu.edu.cn}
\address[mymainaddress]{School of Mathematical Science \& Institute of Mathematics, Nanjing Normal University and Key Laboratory of Ministry of Education Numerical Simulation of Large Scale Complex Systems,
Nanjing, 210023, China.}

\author[add2]{Gregory Gutin}
\ead{G.Gutin@rhul.ac.uk}
\address[add2]{Royal Holloway, University of London, Egham, Surrey, TW20 0EX, UK,}

\author[add3]{Yaping Mao}
\ead{maoyaping@ymail.com}
\address[add3]{School of Mathematics and Statistics and Academy of Plateau Science and Sustainability, Qinghai Normal University, Xining, Qinghai 810008, China}

\author[add4]{Donglei Du}
\ead{ddu@unb.ca}
\address[add4]{Faculty of Management, University of New Brunswick, Fredericton, New Brunswick, E3B9Y2, Canada,}

\author[mymainaddress]{Xiaoyan Zhang\corref{mycorrespondingauthor}}
\cortext[mycorrespondingauthor]{Corresponding author}
\ead{zhangxiaoyan@njnu.edu.cn}

\begin{abstract}
Given a natural number $k\ge 2$, we consider the $k$-submodular cover problem ($k$-SC). The objective is to find a minimum cost subset of a ground set $\mathcal{X}$ subject to the value of a $k$-submodular utility function being at least a certain predetermined value $\tau$. For this problem, we design a bicriteria algorithm with a cost at most $O(1/\epsilon)$ times the optimal value, while the utility is at least $(1-\epsilon)\tau/r$, where $r$ depends on the monotonicity of $g$.
\end{abstract}







\end{frontmatter}







\section{Introduction}
Let $\mathbb{R}_+$  denote the set of all non-negative real numbers. Given a ground set $\mathcal{X}$, a set function $g:2^\mathcal{X}\rightarrow \mathbb{R}_+$ is called a submodular function if for any subset $S, T\subseteq \mathcal{X}$, $g(S)+g(T)\geq g(S\cup T)+g(S\cap T)$. The study of submodular set functions has attracted much attention due to their applications to problems arising in various areas, such as influence maximization \cite{Kempe2003, Stoica2020} and recommendation system \cite{Serbos2017, Wang2021}.
A natural extension of a submodular function is a $k$-submodular function, where an argument of a $k$-submodular function is a collection of $k$ disjoint subsets $\textbf{s}=(S_1,S_2,\dots, S_k)$, rather than a single set in the traditional subnmodular set function. Therefore, a submodular function is a special case of $k$-submodular function when $k=1$. Huber and Kolmogorov \cite{Huber2012} introduced the concept of a $k$-submodular function and studied minimization of $k$-submodular functions. Due to its universality and diminishing return property, $k$-submodular functions have many applications in real world, such as $k$-topic influence maximization problem \cite{Pham2022b, Qian2017, Rafiey2020} and sensor placement with $k$ types of sensors \cite{Ene2022, Ohsaka2015}.

Maximization of $k$-submodular functions under various constraints have also been studied in recent years. For the unconstrained problem, Ward and \v{Z}ivn\'{y} \cite{Ward2016} first designed a deterministic $1/3$-approximation algorithm and a randomized $1/(1+a)$ approximation algorithm where $a=\max\{1, \sqrt{(k-1)/4}\}$.
Later, Iwata et al. \cite{Iwata2016} provided a $1/2$ approximation randomized algorithm for the non-monotone case and a $\frac{k}{2k-1}$-approximation algorithm for the monotone case.  Using the framework of \cite{Iwata2016}, Oshima \cite{Oshima2021} improved the approximation ratio to $\frac{k^2+1}{2k^2+1}$ for the non-monotone case if $k\geq 3$.
Ohsaka and Yoshida \cite{Ohsaka2015} studied the total size (i.e., $|\cup_{i\in[k]}S_i|\leq B$ where $B\geq0$) and individual size  constraints (i.e., $|S_i|\leq B_i$ where $B_i\geq0$ $\forall i\in [k]$) when a $k$-submodular function is monotone, achieving $\frac{1}{2}$ and $\frac{1}{3}$-approximation ratios, respectively.
Further, Xiao et al. \cite{Xiao2022} studied maximization of non-monotone $k$-submodular functions under individual size constraints and presented an algorithm with approximate ratio $\frac{1}{\max_{i\in [k]}B_i+4}$.
Sakaue \cite{Sakaue2017} studied maximization of a monotone $k$-submodular function under a matroid constraint and provided a $\frac{1}{2}$-approximation algorithm. For a knapsack constraint, Tang \cite{Tang2021} gave a 0.4-approximation algorithm for the monotone case. This result was improved to $\frac{1}{2}-\epsilon$ \cite{Wang2021} using the multilinear extension method.

In applications, due to the increasing scale in data volumes and the nature of problems, it is impractical to store all data and hence necessary to design streaming algorithms that could process data as it arrives. Streaming algorithms have the advantage of using limited memory and less time, but still could return a reasonable solution. Pham et al. \cite{Pham2022a} studied streaming algorithms for the $k$-submodular maximization under a knapsack constraint, achieving $\frac{1}{4}-\epsilon$ and $\frac{1}{5}-\epsilon$ approximation ratios for monotone and non-monotone cases, respectively. Ene and Nguyen \cite{Ene2022} extended the primal dual method for streaming submodular maximization under a matching constraint \cite{Levin2021} to the monotone $k$-submodular maximization under an individual size constraint and a partition matroid constraint; the approximation ratios are both at least 0.25; and the approximation ratio asymptotically increases to 0.2953 and 0.3178 as $B$ tends to infinity where $B=\min_{i\in[k]}B_i$.
Spaeh, Ene and Nguyen \cite{Spaeh2023} further used the local search method to design algorithms for the non-monotone $k$-submodular maximization problem, achieving at least 0.125 and 0.175 approximation ratios for an individual size constraint and a partition matroid constraint, respectively. The approximation ratio would asymptotically increase to 0.1589 and 0.1921 as $B$ tends to infinity.

So far, many researchers have developed algorithms for maximizing $k$-submodular problems under various constraints. These problems  model many real-life applications.
Inspired by these problems, Pham et al. \cite{NHPham2022} proposed the $k$-submodular cover ($k$-SC) problem. This problem  model the problem of influence threshold with $k$ topics \cite{Pham2022b}. Given a utility function $g$, the $k$-SC problem is to find a solution such that its utility cost is at least $\tau$ while minimizing the size of the solution. We extend the $k$-submodular cover problem to the weighted $k$-submodular cover problem where the objective function is a weighted cost function. The definition of weighted $k$-Submodular Cover is introduced as follows.

\begin{definition}
\rm We are given a ground set $\mathcal{X}$, a weighted cost function $w:\mathcal{X}\rightarrow \mathbb{R}_+$, a $k$-submodular utility function $g:(k+1)^\mathcal{X}\rightarrow \mathbb{R}_+$, and a utility threshold $\tau$ such that $\tau\leq \max_{\textbf{s}\in (k+1)^\mathcal{X}}g(\textbf{s})$. Let $w(\textbf{x})=\sum_{i\in[k]}\sum_{x\in S_i}w(x),$ where $[k]=\{1,2,\dots,k\}$. Our goal is to find a solution $\textbf{x}\in (k+1)^\mathcal{X}$ such that
\begin{equation}\label{Pro1}
\begin{aligned}
\min \quad & w(\textbf{x})\\
\mbox{s.t.} \quad & g(\textbf{x})\geq \tau\\
& \textbf{x}\in (k+1)^\mathcal{X}.
\end{aligned}
\end{equation}
\end{definition}

The Submodular Cover is a special case of $k$-SC; this problem and its variants have numerous literature due to their wide applications in real world, e.g., \cite{Crawford2023,Mirzasoleiman2015, Mirzasoleiman2016, Norouzi2016, Wang2023,Wolsey1982}. Since Submodular Cover is well-known to be NP-hard, $k$-SC is also NP-hard. Some researchers designed algorithms from the bicriteria perspective for the submodular cover problem because of the method they used or the nature of submodular function.
An algorithm is called an {\em $(\alpha,\beta)$-bicriteria algorithm} if it is satisfies $w(S)\leq \alpha w(V)$ and $g(S)\geq \beta \tau$, where $S$ is the output solution of the algorithm and $V$ is the optimal solution. When $g$ is a monotone DR-submodular function, where the domain of a DR-submodular function is an integer lattice, Soma and Yoshida \cite{Soma2015} presented a threshold based greedy algorithm, and proved that the algorithm has $((1+3\epsilon)\rho(1+\ln\frac{d}{\beta}), 1-\delta)$-bicriteria approximation ratio where $0<\epsilon, \delta <1$, $\rho$ is the curvature of $w$ (for its definition, see \cite{Soma2015}), $d$ and $\beta$ are instance dependent parameters. For the streaming setting, \cite{Norouzi2016} designed a $(\frac{2}{\epsilon}, 1-\epsilon)$-bicriteria algorithm for the submodular cover problem with uniform cost (i.e. $w(x)=1, \forall x\in \mathcal{X}$).
If $g$ is non-monotone submodular, using a $\gamma$-approximation algorithm for unconstrained non-monotone submodular maximization problem as the subroutine, Crawford \cite{Crawford2023} designed a one-pass algorithm with $((1+\epsilon)(4/\epsilon^2+1), \gamma(1-\epsilon))$-bicriteria approximate ratio.
Pham et al. \cite{NHPham2022} studied the $k$-Submodular Cover problem with uniform cost, achieving $(\frac{1-\epsilon^2}{2\epsilon},\frac{1-\epsilon}{2})$ approximation ratio.

\textbf{Our Contribution}. In this paper, we consider $k$-SC in the streaming setting. We generalize the previous work \cite{NHPham2022} from the uniform cost to the general non-negative weighted cost. Three algorithms are devised to resolve this problem. The first algorithm is based on the assumption of knowing the optimal solution value. Later, we provide Algorithm \ref{AL2} with streaming twice in the algorithm; such procedure could remove the assumption. We design Algorithm 3 which only streams the data set only once, but still maintains the bicriteria approximation ratio. We also analyze the memory and update time of an element of the algorithms.

\section{Preliminaries}

Given a finite set $\mathcal{X}$ of size $n$ and a positive integer $k$, we define $[k]=\{1,2,\dots, k\}$. Let $(k+1)^{\mathcal{X}}=\{(S_1, S_2\dots S_k)|S_i\subseteq \mathcal{X}, S_i\cap S_j =\emptyset, \text{ for all } i,j \in [k]\}$ be the set of all $k$ disjoint subsets of $\mathcal{X}$. Each subset in $(k+1)^\mathcal{X}$ is called $k$-set. We denote $\textbf{0}$ as $(\emptyset, \emptyset, \dots \emptyset)$ for simplicity. Given a $k$-set $\textbf{s}=(S_1,S_2,\dots, S_k)\in (k+1)^\mathcal{X}$, let $\textbf{s}(x)=i$ if $x\in S_i$ and $i$ is called the $position$ in $\textbf{x}$; $\textbf{s}(x)=0$ if $x\notin S_i$.

For $k$-sets $\textbf{s} = (S_1, S_2\dots S_k)$ and $\textbf{t} = (T_1, T_2\dots T_k)\in (k+1)^\mathcal{X}$, two operations meet and join $\textbf{s}\sqcap \textbf{t}$ and $\textbf{s}\sqcup \textbf{t}$ is defined as following
\[
\textbf{s}\sqcap \textbf{t} = (S_1\cap T_1, S_2\cap T_2,\dots, S_k\cap T_k),
\]
\[
\textbf{s}\sqcup \textbf{t} = (S_1\cup T_1\backslash \bigcup\limits_{i\neq1}(S_i\cup T_i), 
\dots, S_k\cup T_k\backslash \bigcup\limits_{i\neq k}(S_i\cup T_i)).
\]

A function $g: (k+1)^{\mathcal{X}}\rightarrow \mathbb{R}_+$ is {\em $k$-submodular} if $k$-sets $\textbf{s}$ and $\textbf{t}\in (k+1)^{\mathcal{X}}$, we have
\[
g(\textbf{s}) + g(\textbf{t}) \geq g(\textbf{s}\sqcap \textbf{t}) + g(\textbf{s}\sqcup\textbf{t}).
\]
In this paper, we assume that $g$ is non-negative and normalized, i.e., $g(\textbf{s})\geq 0$ for all $\textbf{s}\in (k+1)^\mathcal{X}$ and $g(\textbf{0})=0$. We also assume that an oracle is given to access the $k$-submodular function; the oracle is a black box returning $g(\textbf{s})$ for a given $k$-set $\textbf{s}$.
We write $\textbf{s}\sqsubseteq \textbf{t}$ if and only if $S_i\subseteq T_i$ for each $i\in [k]$. A function $g$ is {\em monotone} if for any $\textbf{s}$ and $\textbf{t}\in (k+1)^\mathcal{X}$ with $\textbf{s}\sqsubseteq \textbf{t}$, we have $g(\textbf{s})\leq g(\textbf{t})$. Let $E(\textbf{s})=S_1\cup S_2 \dots \cup S_k$. Let $(x,i)$ denote a pair where $x\in \mathcal{X}$ and $i$ is the position. Therefore, a $k$-set $\textbf{s}$ can also be written as $\textbf{s}=\{(x_1,i_1),(x_2,i_2),\dots,(x_j,i_j)\}$ where $x_p\in E(\textbf{s})$ and $i_p=\textbf{s}(x_p)$ for $1\leq p\leq j$. Adding an element $x \notin E(\textbf{s})$ to the $S_i$ can be written as $\textbf{s}\sqcup (x,i)$. Define
\[
\Delta_{x,i}g(\textbf{x}) = g(\textbf{s}\sqcup (x,i)) - g(\textbf{s})
\]
as the marginal gain. We introduce some basic but useful properties of a $k$-submodular function which is crucial in the analysis of our algorithm.
From \cite{Ward2016}, a $k$-submodular function $g$ has $orthant$ $submodularity$, that is,
\[
\Delta_{x,i}g(\textbf{s}) \geq \Delta_{x,i}g(\textbf{t})
\]
where $\textbf{s}\sqsubseteq \textbf{t}\in (k+1)^\mathcal{X}$, $x\notin E(\textbf{t})$ and $i\in[k]$. Also, a $k$-submodular function has {\em pairwise monotonicity}, i.e., if $x\notin E(\textbf{s})$, then
\[
\Delta_{x,i}g(\textbf{s}) + \Delta_{x,j}g(\textbf{s}) \geq 0
\]
for any $i, j\in[k]$ with $i\neq j$.

\begin{lemma}\label{lem1}\cite{Tang2021}
\rm Given a $k$-submodular function $g$, we have
\[
g(\textbf{t})\leq g(\textbf{s})+\sum\limits_{x\in E(\textbf{t})\backslash E(\textbf{s})}\Delta_{x,\textbf{t}(x)}g(\textbf{s}).
\]
for any $\textbf{s}$, $\textbf{t}\in (k+1)^{\mathcal{X}}$ such that $\textbf{s}\sqsubseteq \textbf{t}$.
\end{lemma}

Before presenting our algorithm for the $k$-submodular cover problem, we first explain the limitation of finding a feasible solution for Problem (\ref{Pro1}). The $k$-submodular cover problem has a connection with the unconstrained $k$-submodular maximization problem (U$k$-SM). Indeed, recently Crawford \cite{Crawford2023} proved that for every $\epsilon>0$ and for every $(\alpha,\beta)$-bicriteria approximation of the former we could find a $\beta-\epsilon$-approximation for U$k$-SM. Thus, the following holds.
\begin{theorem}\cite{Crawford2023}
For any $\epsilon>0$, a polynomial time $(\alpha,\beta)$-bicriteria algorithm for the $k$-SC can be converted to a $\beta$-approximation algorithm for U$k$-SM in polynomial time.
\end{theorem}


If $g$ is monotone, let $\beta\geq \frac{k+1}{2k}+\epsilon$, then $\mathbb{A}$ could be converted to a ($\frac{k+1}{2k}+\epsilon$)-approximation algorithm for the U$k$-SM problem, while Iwata et al. \cite{Iwata2016} proved that a $(\frac{k+1}{2k}+\epsilon)$-approximation algorithm for U$k$-SM problem would requires an exponential number of queries. Therefore $\beta\leq \frac{k+1}{2k}$ since $\epsilon$ can be arbitrarily small. If $g$ is non-monotone, then by setting $\beta\geq \frac{1}{3}+\epsilon$, we can get $g(\textbf{s})\geq (\frac{1}{3}+\epsilon)OPT$, but  Ward and \v{Z}ivn\'{y} \cite{Ward2016} gave a deterministic greedy algorithm for non-monotone U$k$-SM, achieving $\frac{1}{3}$-approximation ratio. They also gave a tight instance that could achieve $\frac{1}{3}$-approximation ratio, therefore $\beta\leq \frac{1}{3}$ for any instances since our algorithm is based on deterministic greedy strategy.


\section{Streaming Algorithm with Known Guess of Optimal Value}

In this section, we present our greedy threshold streaming algorithm and its theoretical analysis. This algorithm is Algorithm \ref{AL1}.
%

\subsection{The streaming algorithm}

\floatname{algorithm}{Algorithm}
\renewcommand{\algorithmicrequire}{\textbf{Input:}}
\renewcommand{\algorithmicensure}{\textbf{Output:}}

\begin{algorithm}[htb]
\textbf{Input:} $\tau$, $\overline{w(\textbf{v})}$ and $\epsilon\in (0,1]$
\caption{\textbf{}}
\label{AL1}
        \begin{algorithmic}[1]
            \STATE Set $\textbf{x}:=\textbf{0}$, $\theta:=\epsilon\tau/\overline{w(\textbf{v})}$
            \STATE Set $A:=\frac{3-\epsilon}{2\epsilon}\overline{w(\textbf{v})}$ if $g$ is monotone; otherwise, $A:=\frac{4-\epsilon}{3\epsilon}\overline{w(\textbf{v})}$
            \FOR {$x\in \mathcal{X}$}
                    \STATE $i':=\arg\max_{i\in[k]} g((x,i))$
                    \IF {$w(x)\leq A$ and $g((x,i'))\geq \tau$}\label{alLine1}
                        \STATE $\textbf{x} := (x, i')$
                    \ELSE
                        \STATE $i':=\arg\max_{i\in[k]} \Delta_{x,i}g(\textbf{x})$
                        \IF {$\frac{\Delta_{x,i'}g(\textbf{x})}{w(x)}\geq \theta$ and $w(\textbf{x})+w(x)\leq A$}\label{alLine2}
                            \STATE $\textbf{x} := \textbf{x}\sqcup (x,i')$
                        \ENDIF
                    \ENDIF
            \ENDFOR
            \RETURN {$\textbf{x}$}
        \end{algorithmic}
\end{algorithm}

Let $\textbf{v}$ be an optimal solution of the Problem (\ref{Pro1}), that is,
\[
w(\textbf{v})=\arg\min\{w(\textbf{s})|g(\textbf{s})\geq \tau, \textbf{s}\in (k+1)^\mathcal{X}\}.
\]

Assume that a guessed value $\overline{w(\textbf{v})}$ is known and satisfies $w(\textbf{v})\leq \overline{w(\textbf{v})}$.
Algorithm \ref{AL1} takes the guessed value $\overline{w(\textbf{v})}$, $\epsilon\in (0,1]$ and $r$ as inputs. At the beginning of Algorithm \ref{AL1}, it initialises an empty $k$-set $\textbf{x}$ which pairs will be added in the algorithm process.
We denote $\theta$ as the threshold value and $A$ as the upper bound cost of $\textbf{x}$ which means at the end of Algorithm \ref{AL1}, $w(\textbf{x})$ is at most $A$.
We use $\theta=\epsilon\tau/\overline{w(\textbf{v})}$, $A=\frac{3-\epsilon}{2\epsilon}\overline{w(\textbf{v})}$ and $\frac{4-\epsilon}{3\epsilon}\overline{w(\textbf{v})}$ for monotone and non-monotone cases, respectively.
An element $x$ is called $big$ if it satisfies $w(x)\leq A$ and $\max_{i\in[k]}g((x,i))\geq \tau$. When a new element arrives, Algorithm \ref{AL1} first checks whether $x$ is a $big$ element or not. If $x$ is $big$, then the algorithm sets $\textbf{x}=(x,i')$ where $i'=\arg\max_{i\in[k]}g((x,i))$. If it is not $big$, the algorithm finds a position $i'\in [k]$ that maximizes $\Delta_{x,i}g((x,i))$ and inserts $(x, i')$ to the current $k$-set $\textbf{x}$ if $\frac{\Delta_{x,i'}g(\textbf{x})}{w(x)}\geq \theta$ and $w(\textbf{x})+w(x)\leq A$ are satisfied. After all elements of $\mathcal{X}$ are scanned over,
Algorithm \ref{AL1} returns $\textbf{x}$ as the final solution.

\subsection{Analysis of Algorithm \ref{AL1}}

In the following, we give the performance guarantee of Algorithm \ref{AL1} for both monotone and non-monotone cases.
\begin{lemma}\label{lem2}
\rm Suppose that the output solution $\textbf{x}$ of Algorithm \ref{AL1}  contains a $big$ element $x_b$. If $g$ is monotone, then $w(\textbf{x})\leq \frac{3-\epsilon}{2\epsilon}\overline{w(\textbf{v})}$ and $g(\textbf{x})\geq \tau$. If $g$ is non-monotone, then $w(\textbf{x})\leq \frac{4-\epsilon}{3\epsilon}\overline{w(\textbf{v})}$ and $g(\textbf{x})\geq \tau$.
\end{lemma}
\begin{proof}
Assume that $g$ is monotone. Since $E(\textbf{x})$ contains a $big$ element $x_b$, then from Line \ref{alLine1} of Algorithm \ref{AL1}, we have
\[
w(x_b)\leq A \text{ and } \max_{i\in[k]}g((x_b, i)) \geq \tau.
\]
That is, the single element satisfies the cover constraint. Moreover, $w(\textbf{x})\leq A$ also holds according to the condition in Line \ref{alLine2}. Then we have
\begin{equation}\label{eqsup}
w(x_b)\leq w(\textbf{x}) \leq A = \frac{3-\epsilon}{2\epsilon}\overline{w(\textbf{v})}.
\end{equation}
The inequality holds due to the fact that $x_b\in E(\textbf{x})$.
And by monotonicity of $g$ we obtain
\[
g(\textbf{x})\geq \max_{i\in[k]}g((x_b,i))\geq \tau.
\]
For the non-monotone case, since the $k$-submodular function has pairwise monotonicity property, $\Delta_{x,i}g(\textbf{s})+\Delta_{x,j}g(\textbf{s})\geq 0$ for any $\textbf{s}\in (k+1)^\mathcal{X}$ and $x\notin E(\textbf{s})$ with $i\neq j$,  there is at most one position $i$ such that $\Delta_{x,i}g(\textbf{s})<0$.
By the greedy rule of Algorithm \ref{AL1}, $g(\textbf{x})$ is non-decreasing during the algorithm process. Hence, we obtain
\[
w(\textbf{x}) \leq \frac{4-\epsilon}{3\epsilon}\overline{w(\textbf{v})} \text{ and }g(\textbf{x})\geq \tau.
\]
\end{proof}

Therefore, we now assume that $\mathcal{X}$ does not contain any $big$ element. Since the upper bound of $w(\textbf{x})$ is given in the algorithm, we only need to analyze the value of $g(\textbf{x})$. First, we introduce additional notation as follows.
\begin{itemize}
    \item Let $|E(\textbf{x})|=m$.
    \item Let $(x_p,i_p)$ be the $p$-th pair added to $\textbf{x}$ in the algorithm where $p=1,\dots,m$.
    \item Let $\textbf{x}^p=\{(x_1,i_1), (x_2, i_2),\dots, (x_p, i_p)\}$, i.e., the first $p$ pairs added to $\textbf{x}$. Specifically, let $\textbf{x}^0=\textbf{0}$ and $\textbf{x}^m=\textbf{x}$.
    \item $\textbf{v}^p=(\textbf{v}\sqcup \textbf{x}^p)\sqcup \textbf{x}^p$ where $p=0,1,\dots,m$.
    \item $\textbf{v}^{p-\frac{1}{2}}=(\textbf{v}\sqcup \textbf{x}^p)\sqcup \textbf{x}^{p-1}$ where $p=1,\dots,m$.
    \item For any $x\in \mathcal{X}$, let $\textbf{x}^{m_x}$ denote $\textbf{x}$ when $x$ is considered in the algorithm.
\end{itemize}
By the definitions of $\{\textbf{x}^p\}_{p=0}^m$, $\{\textbf{v}^p\}_{p=0}^m$ and $\{\textbf{v}^{p-\frac{1}{2}}\}_{p=1}^m$, it is easy to verify that $\textbf{x}^{p-1} \sqsubseteq \textbf{v}^{p-\frac{1}{2}}$, $\textbf{v}^{p-\frac{1}{2}}\sqsubseteq \textbf{v}^p$, $\textbf{x}^p \sqsubseteq \textbf{v}^p$ and $\textbf{x}^{m_x}\sqsubseteq \textbf{x}$.

\begin{lemma}\label{lem3}\cite{Pham2022a}
\rm Given an optimal solution $\textbf{v}$ for Problem (\ref{Pro1}), we have that
\begin{enumerate}
  \item[(a)] if $g$ is monotone, then
  \begin{equation}
  g(\textbf{v})-g(\textbf{v}^m)\leq g(\textbf{x}).
  \end{equation}
  \item[(b)] if $g$ is non-monotone, then
  \begin{equation}
  g(\textbf{v})-g(\textbf{v}^m)\leq 2g(\textbf{x}).
  \end{equation}
\end{enumerate}
\end{lemma}

\begin{lemma}\label{lem4}
\rm If $\mathcal{X}$ does not contain any $big$ element, then $g(\textbf{x})\geq \frac{(1-\epsilon)\tau}{2}$ if $g$ is monotone and $g(\textbf{x})\geq \frac{(1-\epsilon)\tau}{3}$ if $g$ is non-monotone.
\end{lemma}
\begin{proof}
Suppose that $g$ is monotone. By orthant submodularity of $g$, Lemma \ref{lem1} and Lemma \ref{lem3}, we have
\begin{equation}\label{eq3}
\begin{split}
\tau - g(\textbf{x}) &\leq g(\textbf{v})- g(\textbf{x})\\
                       &= g(\textbf{v})-g(\textbf{v}^m)+g(\textbf{v}^m)-g(\textbf{x})\\
                       &\leq g(\textbf{x})+\sum_{x\in E(\textbf{v})\backslash E(\textbf{x})}\Delta_{x, \textbf{v}(x)}g(\textbf{x}).
\end{split}
\end{equation}

For any element $x\in E(\textbf{v})\backslash E(\textbf{x})$, since it is not inserted into the $k$-set $\textbf{x}$ in the algorithm process, it does not pass the condition in Line \ref{alLine2}. An element $x\in E(\textbf{v})\backslash E(\textbf{x})$ is called $good$, if it satisfies $\max_{i\in[k]}\frac{\Delta_{x,i}g(\textbf{x}^{m_x})}{w(x)}< \theta$. Otherwise, an element $x\in E(\textbf{v})\backslash E(\textbf{x})$ is $bad$, if it satisfies
\[
\max_{i\in[k]}\frac{\Delta_{x,i}g(\textbf{x}^{m_x})}{w(x)}\geq \theta \text{ and } w(\textbf{x}^{m_x})+w(x)> A  .
\]

We prove this lemma in two cases. First if the set $E(\textbf{v})\backslash E(\textbf{x})$ only contains $good$ elements, then we obtain
\begin{equation}\label{eq4}
\begin{split}
\sum\limits_{x\in E(\textbf{v})\backslash E(\textbf{x})}\Delta_{x,\textbf{v}(x)}g(\textbf{x}) &\leq \sum\limits_{x\in E(\textbf{v})\backslash E(\textbf{x})}\Delta_{x,\textbf{v}(x)}g(\textbf{x}^{m_x})\\
    & = \sum\limits_{x\in E(\textbf{v})\backslash E(\textbf{x})}\frac{\Delta_{x,\textbf{v}(x)}g(\textbf{x}^{m_x})}{w(x)}w(x)\\
    & \leq \sum\limits_{x\in E(\textbf{v})\backslash E(\textbf{x})}\max_{i\in[k]}\frac{\Delta_{x,i}g(\textbf{x}^{m_x})}{w(x)}w(x)\\
    & \leq \sum\limits_{x\in E(\textbf{v})\backslash E(\textbf{x})}\theta w(x)\\
    & \leq \theta w(\textbf{v}).
\end{split}
\end{equation}
The first inequality holds from orthant submodularity of $g$; and the third inequality is true because  $x$ is $good$ element. Combing the inequalities (\ref{eq3}) and (\ref{eq4}) and $w(\textbf{v})\leq \overline{w(\textbf{v})}$, we derive
\begin{equation}\label{eq5}
g(\textbf{x})\geq \frac{\tau-\theta w(\textbf{v})}{2} \geq \frac{(1-\epsilon)\tau}{2}.
\end{equation}

Next, we consider the case that there exist $bad$ elements in $E(\textbf{v})\backslash E(\textbf{x})$. Suppose that $x$ is a $bad$ element and $i'=\arg\max_{i\in[k]}\Delta_{x,i}g(\textbf{x}^{m_x})$. Then we have
\begin{align*}
g(\textbf{x}^{m_x}\sqcup (x,i')) &= g(\textbf{x}^{m_x})+ g(\textbf{x}^{m_x}\sqcup (x,i'))  -g(\textbf{x}^{m_x})\\                             &=\sum\limits_{p=1}^{m_x}[g(\textbf{x}^p)-g(\textbf{x}^{p-1})]+g(\textbf{x}^{m_x}\sqcup (x,i'))-g(\textbf{x}^{m_x})\\
                              &=\sum\limits_{p=1}^{m_x}\frac{\Delta_{x_p,i_p}g(\textbf{x}^{p-1})}{w(x_p)}w(x_p)+
                              \frac{\Delta_{x,i'}g(\textbf{x}^{m_x})}{w(x)}w(x)\\
                              &\geq \sum\limits_{p=1}^{m_x}\theta w(x_p)+\theta w(x)\\
                              &= \theta w(\textbf{x}^{m_x})+\theta w(x)\\
                              &= \theta w(\textbf{x}^{m_x}\sqcup (x,i')) \\
                              &\geq \theta A.
\end{align*}

The first inequality holds due to monotonicity of $g$; the second inequality follows from the fact that when every pair $(x_p,i_p)$ that is inserted into the current $k$-set $\textbf{x}$ should pass the condition in Line \ref{alLine2} of the algorithm and the element $x$ is $bad$. Therefore, we obtain
\begin{equation}\label{eq6}
\theta A\leq g(\textbf{x}^{m_x}\sqcup (x,i')) \leq g(\textbf{x}^{m_x}) + g((x,i'))\leq g(\textbf{x})+\max_{i\in[k]}g((x,i))
\end{equation}


We claim that $\max_{i\in[k]}g((x,i))<\tau$. 
Otherwise, if $\max_{i\in[k]}g((x,i))\geq\tau$, since $x\in E(\textbf{v})\backslash E(\textbf{x})$ we have $w(x)\leq w(\textbf{v})\leq A$, meaning that $x$ is a $big$ element which contradicts that $\mathcal{X}$ contain no $big$ element. Therefore, $\max_{i\in[k]}g((x,i))<\tau$ must hold.
%
Combing it with the inequality (\ref{eq6}) and the definitions of $\theta$ and $A$, we have
\[
g(\textbf{x})\geq \theta A-\tau =\frac{(1-\epsilon)\tau}{2}.
\]
Combing the two cases, we derive that $g(\textbf{x})\geq \frac{(1-\epsilon)\tau}{2}$ if $g$ is monotone. For the non-monotone case, analogously, we have $g(\textbf{x})\geq \frac{(1-\epsilon)\tau}{3}$.
\end{proof}
\qed
\begin{theorem}\label{thm1}
\rm If $g$ is a monotone $k$-submodular function, Algorithm \ref{AL1} returns a solution $\textbf{x}$ such that $w(\textbf{x})\leq \frac{3-\epsilon}{2\epsilon}\overline{w(\textbf{v})}$ and $g(\textbf{x})\geq \frac{(1-\epsilon)\tau}{2}$. If $g$ is a non-monotone $k$-submodular function, then $w(\textbf{x})\leq \frac{4-\epsilon}{3\epsilon}\overline{w(\textbf{v})}$ and $g(\textbf{x})\geq \frac{(1-\epsilon)\tau}{3}$.
\end{theorem}
\begin{proof}
Since $A$ is an upper bound on the cost of $w(\textbf{x})$, if $g$ is monotone, we have
\[
w(\textbf{x})\leq \frac{3-\epsilon}{2\epsilon}\overline{w(\textbf{v})}.
\]
Also, for the non-monotone case, we have $w(\textbf{x})\leq \frac{4-\epsilon}{3\epsilon}\overline{w(\textbf{v})}$. In Lemma \ref{lem2} and Lemma \ref{lem4}, we have shown that $g(\textbf{x})\geq \frac{(1-\epsilon)\tau}{2}$ and $g(\textbf{x})\geq \frac{(1-\epsilon)\tau}{3}$ for monotone and non-monotone case, respectively.
\end{proof}
\qed

\section{The 2-pass streaming algorithm}

We assume that the optimal value $w(\textbf{v})$ is known in Algorithm \ref{AL1}. In this section,
we present an algorithm that  removes this assumption by streaming the ground set $\mathcal{X}$ twice and provide a bicriteria approximation ratio. This algorithm is Algorithm \ref{AL2}.

\floatname{algorithm}{Algorithm}
\renewcommand{\algorithmicrequire}{\textbf{Input:}}
\renewcommand{\algorithmicensure}{\textbf{Output:}}

\begin{algorithm}[htb]
\caption{\textbf{}}
\label{AL2}
        \begin{algorithmic}[1]
        \REQUIRE $\tau$, $\epsilon$
            \STATE Set $w_{min}:=+\infty$ and $w_{\max}:=-\infty$
            \FOR {$x\in \mathcal{X}$}
                \IF {$w(x)<w_{min}$}
                    \STATE $w_{min}:=w(x)$
                \ELSIF {$w(x)>w_{\max}$}
                    \STATE $w_{\max}:=w(x)$
                \ENDIF
            \ENDFOR
            \STATE Construct the guess set $\Lambda$ according to $w_{min}$ and $w_{\max}$.
            \STATE Initialize inputs ($\tau$, $\lambda_j, \epsilon)$ of Algorithm \ref{AL1} for all $\lambda_j\in \Lambda$.
            \FOR {$x\in \mathcal{X}$}
                \STATE Pass $x$ to all inputs ($\tau$, $\lambda_j, \epsilon)$ of Algorithm \ref{AL1}
            \ENDFOR
            \STATE Let $\textbf{x}_j$ denote the output solution ($\tau$, $\lambda_j, \epsilon)$ of Algorithm \ref{AL1}
            \STATE Let $H:=\{\textbf{x}_j|g(\textbf{x}_j)\geq \frac{(1-\epsilon)\tau}{r},j=0,1,\dots,l\}$ where $r=2$ if $g$ is monotone; otherwise, $r=3$
            \STATE $\textbf{x}:=\arg\min_{\textbf{x}\in H}w(\textbf{x})$
            \RETURN {$\textbf{x}$}
        \end{algorithmic}
\end{algorithm}

Algorithm \ref{AL2} first streams the ground set $\mathcal{X}$ in order to find $w_{min}$ and $w_{\max},$ where  $w_{min}=\min_{x\in\mathcal{X}}w(x)$ and $w_{\max}=\max_{x\in\mathcal{X}}w(x)$. Let $\gamma = \frac{w_{min}}{w_{\max}}$ and $l=\frac{\ln \frac{\gamma}{n}}{\ln(1-\epsilon)}$. Since
$w_{min}\leq w(\textbf{v}) \leq nw_{\max},$ we  construct the guessed set $\Lambda$ in which every item in $\Lambda$ is a guess of $w(\textbf{v})$. Here $\Lambda$ is defined as
\[
\Lambda =\{(1-\epsilon)^jnw_{\max}|j=0,1,\dots,l\}.
\]
Let $\lambda_j=(1-\epsilon)^jnw_{\max}$. For each $\lambda_j\in \Lambda$, Algorithm \ref{AL2} initializes an input of Algorithm \ref{AL1}. Denote $\textbf{x}_j$ as the output solution of input $(\tau, \lambda_j,\epsilon)$ of Algorithm \ref{AL1}.
Let the threshold value corresponding to $\lambda_j$ be $\theta_j=\frac{\epsilon\tau}{\lambda_j}$.
Let the upper bound cost $A_j=\frac{3-\epsilon}{2\epsilon}\lambda_j$ or $\frac{4-\epsilon}{3\epsilon}\lambda_j$ depending on the whether $g$ is monotonic or not.

In the second streaming over $\mathcal{X}$, each arriving element $x$ would pass all inputs of Algorithm \ref{AL1} to construct the corresponding solution $\textbf{x}_j$ until the streaming is over. Let $H=\{\textbf{x}_j|g(\textbf{x}_j)\geq \frac{(1-\epsilon)\tau}{r},j=0,1,\dots,l\}$ where $r$ depends on the monotonicity of $g$. Algorithm \ref{AL2} outputs its solution $\textbf{x}$ as $\arg\min_{\textbf{x}\in H}\{w(\textbf{x})\}$.
\begin{theorem}\label{thm2}
\rm Algorithm \ref{AL2} is $(\frac{3-\epsilon}{2\epsilon(1-\epsilon)}, \frac{1-\epsilon}{2})$ and $(\frac{4-\epsilon}{3\epsilon(1-\epsilon)}, \frac{1-\epsilon}{3})$-bicriteria approximation for monotone and non-monotone cases, respectively. The algorithm uses at most $O(n\frac{\ln \frac{\gamma}{n}}{\ln(1-\epsilon)})$ memory and makes at most $O(k\frac{\ln \frac{\gamma}{n}}{\ln(1-\epsilon)})$ queries per element.
\end{theorem}

\begin{proof}
Suppose that $g$ is monotone. Instead of analyzing the output $\textbf{x}$, we consider a specific $\lambda_q$ here. Suppose that $\lambda_q$ satisfies $\lambda_{q+1}\leq w(\textbf{v}) \leq \lambda_q$. Then, by Theorem \ref{thm1}
\[
w(\textbf{x}_q)\leq A_q\leq \frac{3-\epsilon}{2\epsilon}\lambda_q=\frac{3-\epsilon}{2\epsilon}\frac{1}{1-\epsilon}(1-\epsilon)^{q+1}nw_{\max}\leq \frac{3-\epsilon}{2\epsilon(1-\epsilon)}w(\textbf{v}).
\]
and $g(\textbf{x}_q)\geq \frac{(1-\epsilon)\tau}{2}$ holds if $g$ is monotone. This means that $H$ is not empty since $\textbf{x}_q$ is a candidate solution. Also, for the non-monotone case, by a similar argument, we  obtain
\[
w(\textbf{x}_q)\leq \frac{4-\epsilon}{3\epsilon(1-\epsilon)}w(\textbf{v}).
\]
and $g(\textbf{x}_q)\geq \frac{(1-\epsilon)\tau}{3}$ holds.

The only thing left is to analyze the memory complexity and time complexity of the algorithm. Since the guess set $\Lambda$ has at most $l=\frac{\ln \frac{\gamma}{n}}{\ln(1-\epsilon)}$ guesses of $w(\textbf{v})$. Hence, the algorithm would maintain at most $l$ running inputs, where each input contains at most $n$ elements. So the memory is at most $O(n\frac{\ln \frac{\gamma}{n}}{\ln(1-\epsilon)})$. The algorithm uses $O(k)$ queries for each element, and hence at most $O(k\frac{\ln \frac{\gamma}{n}}{\ln(1-\epsilon)})$ queries per element.
\end{proof}

\section{One-pass streaming algorithm}

In this section, we present a bicriteria approximation algorithm for Problem \ref{Pro1} which takes only one pass over the ground set $\mathcal{X}$. The one pass algorithm is more applicable in the real world where one streaming over the ground set may never store the elements.

\floatname{algorithm}{Algorithm}
\renewcommand{\algorithmicrequire}{\textbf{Input:}}
\renewcommand{\algorithmicensure}{\textbf{Output:}}

\begin{algorithm}[htpb]
\caption{\textbf{Single-pass Algorithm}}
\label{AL3}
        \begin{algorithmic}[1]
        \REQUIRE $\tau$, $B$ and $\epsilon$.
            \STATE Set $\textbf{x}_j:=\textbf{0}$, $\forall j\in\mathbb{N}$
            \STATE Set $L:= -\infty$, $U:=B$.
            \FOR {$x\in \mathcal{X}$}
            \STATE $i':=\arg\max_{i\in[k]}g((x,i))$
            \IF {$\frac{g((x,i'))}{w(x)}>\frac{\epsilon\tau}{L}$}\label{alLine3}
                \STATE $L:=\frac{\epsilon\tau w(x)}{g((x,i'))}$
            \ENDIF
            \STATE Let $\Lambda:=\{(1-\epsilon)^jB|L\leq (1-\epsilon)^jB\leq U,j\in \mathbb{N}\}$
            \STATE Construct $\theta_j$, $A_j$ by each $\lambda_j$ in $\Lambda$ and the monotonicity of $g$
            \FOR {$\lambda_j\in \Lambda$}
                \IF {$w(x)\leq A_j$ and $g((x,i'))\geq \tau$}
                    \STATE $\textbf{x}_j := (x, i')$
                \ELSE
                    \STATE $i':=\arg\max_{i\in[k]} \Delta_{x,i}g(\textbf{x}_j)$
                    \IF {$\frac{\Delta_{x,i'}g(\textbf{x}_j)}{w(x)}\geq \theta_j$ and $w(\textbf{x}_j)+w(x)\leq A_j$}
                        \STATE $\textbf{x}_j := \textbf{x}_j\sqcup (x,i')$
                    \ENDIF
                \ENDIF
            \STATE Let $r=2$ if $g$ is monotone; otherwise $r=3$
            \IF {$g(\textbf{x}_j)\geq \frac{(1-\epsilon)\tau}{r}$}
                \STATE $U:= (1-\epsilon)^jB$
            \ENDIF
            \ENDFOR
            \ENDFOR
            \STATE $\textbf{x} = \arg\max\{g(\textbf{x}_j)|\lambda_j\in \Lambda\}$
            \RETURN {$\textbf{x}$}
        \end{algorithmic}
\end{algorithm}

In Algorithm \ref{AL2}, it is easy to get an appropriate guess of the optimal value of the Problem (\ref{Pro1}), but it is hard for single pass algorithm, since it is difficult to efficiently determine a useful upper bound of $w(\textbf{v})$ without seeing the ground set $\mathcal{X}$.
To tackle this difficulty, an upper bound $B$ of $w(\textbf{v})$ is taken as one of the input in \cite{Crawford2023, Norouzi2016}. Based on $B$, we present a dynamic way for guessing the optimal value of $w(\textbf{v})$ in Algorithm \ref{AL3} to reduce the memory and time complexities. The pseudocode for the single-pass streaming is listed in Algorithm \ref{AL3}.

Algorithm \ref{AL3} takes $\tau$, $B$ and $\epsilon$ as inputs. The algorithm first initializes $L =-\infty$ which denotes the lower bound of the guess of $w(\textbf{v})$.
In particular, $L$ is updated to $\epsilon\tau w(x)/g((x,i'))$ if the new arriving element $x$ satisfies that $L> \epsilon\tau w(x)/g((x,i'))$ where $i'=\arg\max_{i\in[k]}g((x,i))$.
Then Algorithm \ref{AL3} generates the guess set $\Lambda:=\{(1-\epsilon)^jB|L\leq (1-\epsilon)^jB\leq U,j\in \mathbb{Z}_+\}$. 
For each guess value $\lambda_j$ in $\Lambda$, we also use the notation $\theta_j$, $A_j$ and $\textbf{x}_j$ which have the same meanings as in Section 4.
$U$ is updated to $\lambda_j$ if the corresponding $\textbf{x}_j$ satisfies $g(\textbf{x}_j)\geq \frac{(1-\epsilon)\tau}{r}$ where $r$ depends on the monotonicity of $g$. Therefore $U$ is non-increasing in the algorithm process.
Algorithm \ref{AL3} can be viewed as running multiple instances of modified Algorithm \ref{AL2} in parallel for each $\lambda_j\in \Lambda$.
If a guess $\lambda_j$ is larger than $U$, it will be discarded from $\Lambda$ in the rest of algorithm precess. Once the ground set is scanned over the stream, Algorithm \ref{AL3} returns the $\arg\max\{g(\textbf{x}_j)|\lambda_j\in \Lambda\}$ as the solution.

\begin{theorem}\label{thm3}
Algorithm \ref{AL3} is a $(\frac{3-\epsilon}{2\epsilon(1-\epsilon)}, \frac{(1-\epsilon)\tau}{2})$ and a $(\frac{4-\epsilon}{3\epsilon(1-\epsilon)}, \frac{(1-\epsilon)\tau}{3})$-approximation algorithm for monotone and non-monotone cases, respectively. Let $\kappa=\max_{x\in\mathcal{X},i\in[k]}g((x,i))/w(x)$. Algorithm \ref{AL3} uses $\frac{n\ln \frac{\epsilon\tau}{B\kappa}}{\ln (1-\epsilon)}$ memory and requires $\frac{k\ln \frac{\epsilon\tau}{B\kappa}}{\ln (1-\epsilon)}$ queries for a $k$-submodular function.
\end{theorem}

\begin{proof}
For any $(1-\epsilon)^jB$ if $L>(1-\epsilon)^jB$, when a new element $x$ arrives, if $\frac{g((x,i'))}{w(x)}\geq \frac{\epsilon\tau}{(1-\epsilon)^jB}$ holds where $i'=\arg\max_{i\in[k]}g((x,i))$, then we  obtain
\[
\frac{g((x,i'))}{w(x)}\geq \frac{\epsilon\tau}{(1-\epsilon)^jB}\geq \frac{\epsilon\tau}{L}.
\]
It implies that $(1-\epsilon)^jB$ is a new lower bound for the guess of the $w(\textbf{v})$.
Then Algorithm \ref{AL3} sets $L=\frac{\epsilon\tau w(x)}{g((x,i'))}$. Thus $(1-\epsilon)^jB$ will be added to guess set $\Lambda$.
Therefore, the guess $\lambda_q=(1-\epsilon)^qB$ such that $\lambda_{q+1}\leq w(\textbf{v})\leq \lambda_q$ will be added to the guess set during the algorithm process since the assumption is $w(\textbf{v})\leq B$.
At the end of the algorithm, $U$ is at most $(1-\epsilon)^qB$ and the final solution $\textbf{x}$ of Algorithm \ref{AL3} satisfies $g(\textbf{x})\geq \frac{(1-\epsilon)\tau}{r}$. The desired results now follow from Theorem \ref{thm2}.

At the end of the Algorithm \ref{AL3}, $L=\frac{\epsilon\tau}{\kappa}$ and so the guess set $\Lambda$ has at most $\frac{\ln\frac{\epsilon\tau}{B\kappa}}{\ln(1-\epsilon)}$ guess values. Thus Algorithm \ref{AL3} uses at most $\frac{n\ln\frac{\epsilon\tau}{B\kappa}}{\ln(1-\epsilon)}$ memory and at most $\frac{k\ln\frac{\epsilon\tau}{B\kappa}}{\ln(1-\epsilon)}$ queries of a $k$-submodular function oracle.
\end{proof}

\section{Conclusion}

We provide three streaming algorithms for the weighted $k$-Submodular Cover problem with the theoretical guarantee of approximation ratio, memory and number of queries. The approximation ratio is $O(\frac{1}{\epsilon})$ for both monotone and non-monotone case. The one-pass algorithm would use $O(n\ln(\frac{\epsilon\tau}{B\kappa})/(\ln(1-\epsilon)))$ memory and at most $O(k\ln(\frac{\epsilon\tau}{B\kappa})/(\ln(1-\epsilon)))$ queries of the $k$-submodular function oracle. In the unconstrained non-monotone $k$-submodular maximization problem, the best result is near $\frac{1}{2}$.
But in our algorithm, $\beta$ is $\frac{1-\epsilon}{3}$ when the function is non-monotone. In the future, we'd like to investigate how to close this gap without increasing the performance cost too much.




\end{document}